\documentclass[showpacs,aps,superscriptaddress,notitlepage,reprint,longbibliography]{revtex4-1}

\usepackage[dvips]{graphicx}
\usepackage{amsmath,amssymb,amsthm,mathrsfs,amsfonts,dsfont,amscd,keyval}
\usepackage{mathtools,mathrsfs}
\usepackage{yfonts} 
\usepackage{subfigure, epsfig}
\usepackage{titletoc} 
\usepackage[subfigure]{tocloft}
\usepackage{titlesec}

\usepackage{ytableau} \ytableausetup{boxsize = 2pt}
\usepackage{braket}
\usepackage{tensor}
\usepackage{bm}
\usepackage{enumerate,enumitem}
\usepackage{color}
\usepackage{hyperref}
\hypersetup{breaklinks=true}
\hypersetup{colorlinks=true, linkcolor=blue, citecolor=magenta, urlcolor=blue, linktocpage=true}
\usepackage[hyphenbreaks]{breakurl}
\usepackage{tabularx}
\usepackage{times}
\usepackage{multirow}
\usepackage{float}
\usepackage{tcolorbox}

\newcommand{\wt}{\widetilde}

\newtheorem*{theorem*}{Theorem}
\newtheorem*{definition*}{Definition}

\theoremstyle{definition}
\newtheorem{definition}{Definition}[section]

\theoremstyle{plain}
\newtheorem{theorem}[definition]{Theorem}
\newtheorem{lemma}[definition]{Lemma}
\newtheorem{corollary}[definition]{Corollary}

\theoremstyle{remark}

\newcommand{\comments}[1]{}

\DeclareMathOperator{\U}{U}
\DeclareMathOperator{\SU}{SU}

\newsavebox{\mstrut}
\newcommand{\bbra}[1]{%
	\sbox{\mstrut}{\(#1\)}%
	\mathinner{\langle\kern-0.3\ht\mstrut\left\langle{#1}\right|}%
}
\newcommand{\kett}[1]{%
	\sbox{\mstrut}{\(#1\)}%
	\mathinner{\left|{#1}\right\rangle\kern-0.3\ht\mstrut\rangle}%
}

\makeatletter
\def\l@subsubsection#1#2{} 
\makeatother

\begin{document}

\title{Toward {Super-polynomial} Quantum Speedup of Equivariant Quantum Algorithms with SU($d$) Symmetry}

\author{Han Zheng$^{**}$}
\email{hanz98@uchicago.edu}
\affiliation{Department of Statistics, The University of Chicago, Chicago, IL 60637, USA}
\affiliation{DAMTP, Center for Mathematical Sciences, University of Cambridge, Cambridge CB30WA, UK}

\author{Zimu Li$^{**}$}
\email{lizm@mail.sustech.edu.cn}
\affiliation{DAMTP, Center for Mathematical Sciences, University of Cambridge, Cambridge CB30WA, UK}

\author{Sergii Strelchuk}
\email{sergii.strelchuk@cs.ox.ac.uk}
\affiliation{Department of Computer Science, University of Oxford}

\author{Risi Kondor}
\email{risi@cs.uchicago.edu}
\affiliation{Department of Statistics, The University of Chicago, Chicago, IL 60637, USA}
\affiliation{Department of Computer Science, The University of Chicago, Chicago, IL 60637, USA}

\author{Junyu Liu}
\email{junyuliu@pitt.edu}
\affiliation{Pritzker School of Molecular Engineering, The University of Chicago, Chicago, IL 60637, USA}
\affiliation{Chicago Quantum Exchange, Chicago, IL 60637, USA}
\affiliation{Kadanoff Center for Theoretical Physics, The University of Chicago, Chicago, IL 60637, USA}
\affiliation{Department of Computer Science, The University of Pittsburgh, Pittsburgh, PA 15260, USA}


\begin{abstract}
     We introduce a framework of the equivariant convolutional quantum algorithms which is tailored for a number of machine-learning tasks on physical systems with arbitrary SU$(d)$ symmetries. It allows us to enhance a natural model of quantum computation -- permutational quantum computing (PQC) [Quantum Inf. Comput., 10, 470–497 (2010)] -- and define a more powerful model: PQC+. While PQC was shown to be efficiently classically simulatable, we exhibit a problem which can be efficiently solved on PQC+ machine, whereas no classical polynomial time algorithm is known; thus providing evidence against PQC+ being classically simulatable. 
     We further discuss practical quantum machine learning algorithms which can be carried out in the paradigm of PQC+. 
     
\end{abstract}

\maketitle


\section{Introduction}
Symmetry plays a fundamental role in modern physics. Symmetry considerations have also emerged as crucial in several areas of quantum computing, such as equivariant quantum optimization algorithms \cite{meyer2022exploiting, Grinko2022LinearPW,PRXQuantum.3.030341,nguyen2022theory} and covariant quantum error correcting codes \cite{eastin2009restrictions,Zhou2021newperspectives,kong2022near}. One of the most fundamental continuous symmetries on qudits is $\SU(d)$ symmetry, which generalizes U(1) charge sectors into $d-1$-dimensional charge sectors \cite{yoshida:soft,junyu2020chargescrambler}. Many challenging physical problems, such as finding the ground state of quantum many-body Hamiltonians, exhibit a natural global continuous symmetry. This is exploited for example in $S_n$-equivariant convolutional quantum alternating ans"atze ($S_n$-CQA) to find the ground state of 2D frustrated magnets \cite{Zheng2021SpeedingUL}. One might ask if such algorithms potentially imply quantum advantage over conventional classical optimization methods such as tensor networks or quantum neural states (NQS). Quantum advantage has been established on the random circuit sampling \cite{boixo2018characterizing,arute2019quantum}; are we capable of achieving the same in the presence of SU($d$) symmetry?

Elegantly, the celebrated Schur--Weyl duality implies that any global SU($d$) symmetric problem is related to not just the representation theory of SU($d$) but also of the symmetric group $S_n$. More precisely, any SU($d$)-symmetric Hamiltonian can be expressed as an element of the \emph{symmetric group algebra} $\mathbb{C}[S_n]$. Notably, Jordan's permutational quantum computing (PQC) \cite{pqc} provides a framework towards investigating quantum advantage in the presence of global SU($d$) symmetry using $S_n$ representation theory. It takes inspiration from topological quantum computing (TQC) and the Penrose spin network model~\cite{penrose1971angular}. In particular, the permutational quantum polynomial time (PQP) complexity class is conjectured to capture classical intractability by approximating matrix entries of the irreducible representations (irreps) of the symmetric group $S_n$ in polynomial time. However, recently a classical polynomial time algorithm was found to compute these elements \cite{sergii1, sergii2}.

In this paper we investigate the question of quantum speedup in the presence of global SU($d$) symmetry by generalizing the framework of PQC, to what we call PQC+. Instead of approximating matrices elements for $S_n$ irreps by polynomial time quantum circuits, let us consider Hermitian operators taken from $\mathbb{C}[S_n]$. As mentioned above, they can be interpreted as $\SU(d)$-symmetric Hamiltonians by Schur--Weyl duality, whose unitary time evolutions are to be approximated in PQC+. The corresponding matrix entries are often called $S_n$ Fourier coefficients. The best known classical algorithms to compute such coefficients take $\mathcal{O}(n! n^2)$ time~\cite{Clausen_1993, Maslen_1998}. Our main result in this paper is to show that quantum circuits are capable of approximating the matrix entry of any $k$-local SU($d$)-symmetric Hamiltonian in polynomial time $\mathcal{O}\left( tCk^3 n^k \frac{\log (tC kn^k/\epsilon)}{\log \log (tC kn^k/\epsilon)}\right)$ up to error $\epsilon$, where $n$ is the number of qudits. Therefore, a potential \emph{{super-polynomial}} quantum speedup is achieved compared to the classical $S_n$ fast Fourier transform (FFT) with $\mathcal{O}(n! n^2)$ costs. Note that similar to the quantum speedup in Jordan's PQC, the quantum circuits can only approximate \emph{single} matrix entry at a time, while classical $S_n$-FFTs compute the entire matrix in the exact sense. However, in the scenarios where we expect this possible {super-polynomial} speedup to be useful, such as approximating the ground state of quantum many-body Hamiltonians, it is sufficient to approximate only individual matrix entries.

It has been shown that single matrix entries of any element in $S_n$ in a given irrep can be classically simulated \cite{sergii2, sergii1}. In our case with any $k$-local SU($d$)-symmetric Hamiltonian, its unitary time evaluation can be approximated using the truncated Taylor series to the order $K$, resulting in $\mathcal{O}(n^{k K})$ many terms. To reach the given precision $\epsilon$ of the truncated Taylor series, the cut-off value $K$ scales with the system size $n$ by $K = \mathcal{O}(\frac{\log (tC kn^k/\epsilon)}{\log \log (tC kn^k/\epsilon)})$. As a result, compared to the method that utilizes the classical algorithm in \cite{sergii2, sergii1} to compute $\mathcal{O}(n^{k K})$ matrix entries from the truncated Taylor series, the quantum algorithm proposed in this work suggests a \emph{{super-polynomial}} speed-up.

\begin{figure}[ht]
	\centering
	\includegraphics[width=2in]{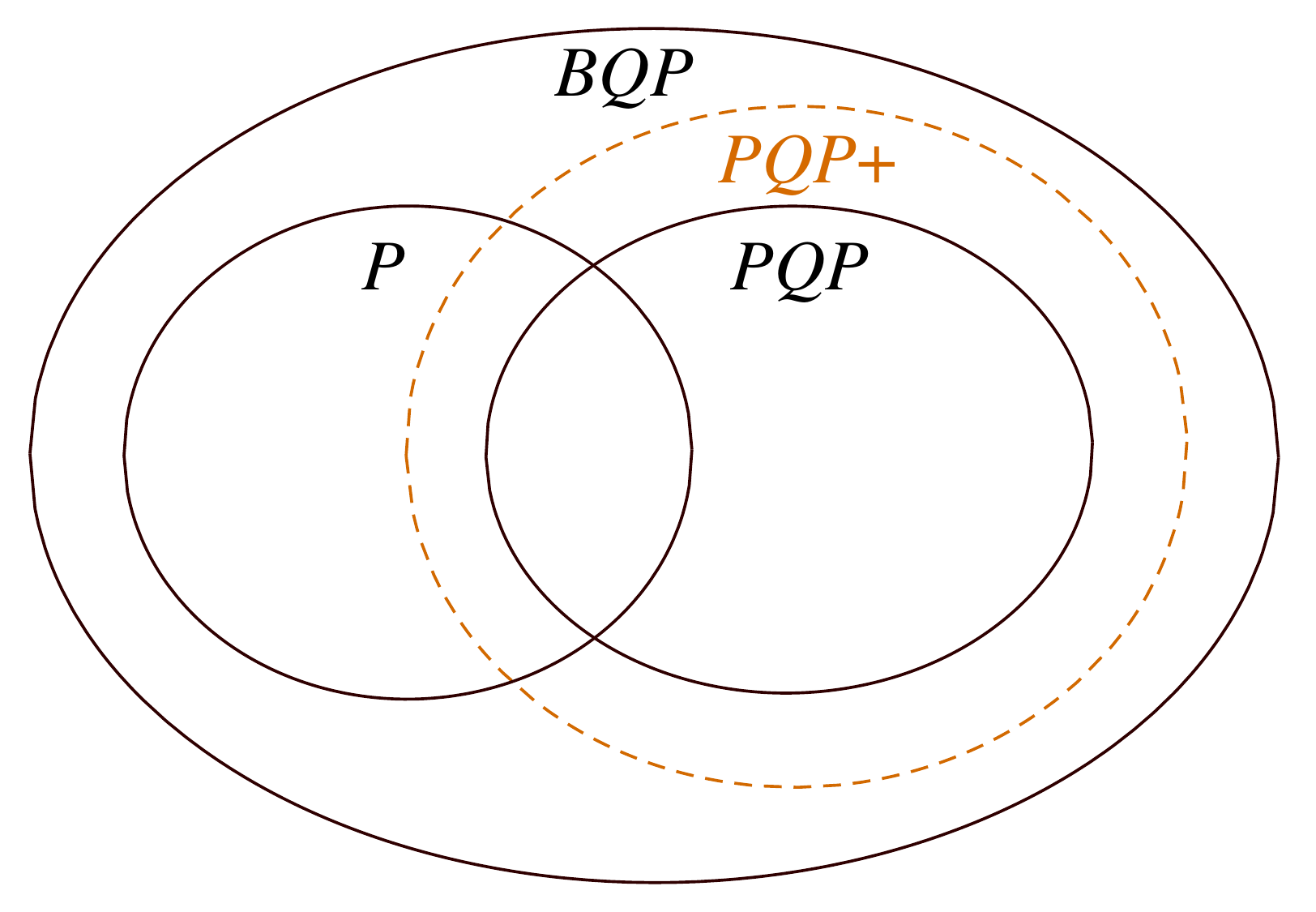}
	\caption{The relationship between complexity classes. It is not known whether PQP+ strictly contains PQP.}
	\label{PQP}
\end{figure}

We demonstrate the significance of the new complexity class PQP+ by showing that matrix multiplication of SU($d$)-symmetric gates has a natural interpretation as $S_n$-equivariant convolution. Equivariance is one of the key properties of classical Convolutional Neural Networks (CNNs), widely used in computer vision. Intuitively, equivariance captures the property that if the input to the neural network is shifted, then its activations translate accordingly. We argue that in quantum circuits, the natural analog of translational equivariance in the Fourier space is the permutation equivariance, for which the corresponding complexity class is PQP+ and offers a generic quantum speedup. The corresponding quantum algorithms are natural candidates for solving physical problems in the presence of global SU($d$) symmetry. By Schur--Weyl duality, it is also desirable to consider quantum algorithms under $S_n$ symmetry. However, as we explain in the following, the subspaces concerned that respect $S_n$ symmetry are always polynomially large in the number $n$ of qudits. Consequently, it should be more careful to argue the potential quantum speedup using the dimensionality in this case. 


\section{Symmetry and Equivariant Convolution}

 Let \( V \) be a \( d \)-dimensional complex Hilbert space with orthonormal basis \( \{e_1,\dots,e_d\} \). We take \( V^{\otimes n} \) to represent the \( n \)-qudit system, and it admits two natural representations: the \textit{tensor product representation} \( \pi_{\SU(d)} \) of \( \SU(d) \) is
\begin{align}
    \pi_{\SU(d)}(g) (e_{i_1} \otimes \cdots \otimes e_{i_n}) \vcentcolon = g \cdot e_{i_1} \otimes \cdots \otimes g \cdot e_{i_n}, 
\end{align}
where \( g \in \SU(d) \) acts transversally on each qudit. The \textit{permutation representation} \( \pi_{S_n} \) of \( S_n \) is
\begin{align}
    \pi_{S_n}(\sigma) (e_{i_1} \otimes \cdots \otimes e_{i_n}) \vcentcolon =
    e_{i_{\sigma^{-1}(1)}} \otimes \cdots \otimes e_{i_{\sigma^{-1}(n)}}.
\end{align}
where \( \sigma \in S_n \) permutes qudits.

Schur--Weyl duality reveals the relationship between these two group actions on the same Hilbert space. This duality is widely used in quantum computation and quantum information \cite{harrow2005PhD,Hayashi2006,childs2007weak,Keppeler2018,Keppeler2018,harrow2016local,Hayashi2023}. In particular, in Quantum Chromodynamics, it was used to decompose the \( n \)-fold tensor product of \( \operatorname{SU}(3) \) representations. In that context, standard Young tableaux are referred to as Weyl-tableaux and labeled by the three isospin numbers (\( u,d,s \)). The underlying Young diagrams containing three rows \( \lambda = (\lambda_1,\lambda_2,\lambda_3) \) are used to denote irreducible representations (irreps) of \( \operatorname{SU}(3) \) \cite{Zeilinger2019}. On the other hand, \( S_n \) irreps are also conventionally indexed by Young diagrams \cite{Sagan01}. Schur--Weyl duality says that irreps of \( \SU(d) \) and \( S_n \) are dual in the following sense, and are denoted by the same Young diagrams \emph{with \( n \) boxes and at most \( d \) rows}.

\begin{theorem*}[Schur--Weyl Duality] 
    The action of \( \SU(d) \) and \( S_n \) on \( V^{\otimes n} \) jointly decomposes the space into irreducible representations of both groups in the form 
    \begin{align}
        V^{\otimes n} = \bigoplus_\lambda W_\lambda\otimes S^\lambda,
    \end{align}
    where \( W_\lambda \) and \( S^\lambda \) denote irreps of \( \SU(d) \) and \( S_n \), respectively, and \( \lambda \) ranges over all Young diagrams of size \( n \) with at most \( d \) rows. Consequently,
    \begin{align}
        & \pi_{\SU(d)} \cong \bigoplus_\mu W_\mu \otimes \operatorname{1}_{m_{\SU(d),\mu}}, \\
        & \pi_{S_n} \cong \bigoplus_\lambda \operatorname{1}_{ m_{S_n,\lambda}} \otimes S^\lambda,
    \end{align}
    where \( m_{\SU(d),\mu} = \dim S^\mu \) and \( m_{S_n,\lambda} = \dim W_\lambda \).
\end{theorem*}

\begin{figure}[H]
    \centering
    \includegraphics[width=3.2in]{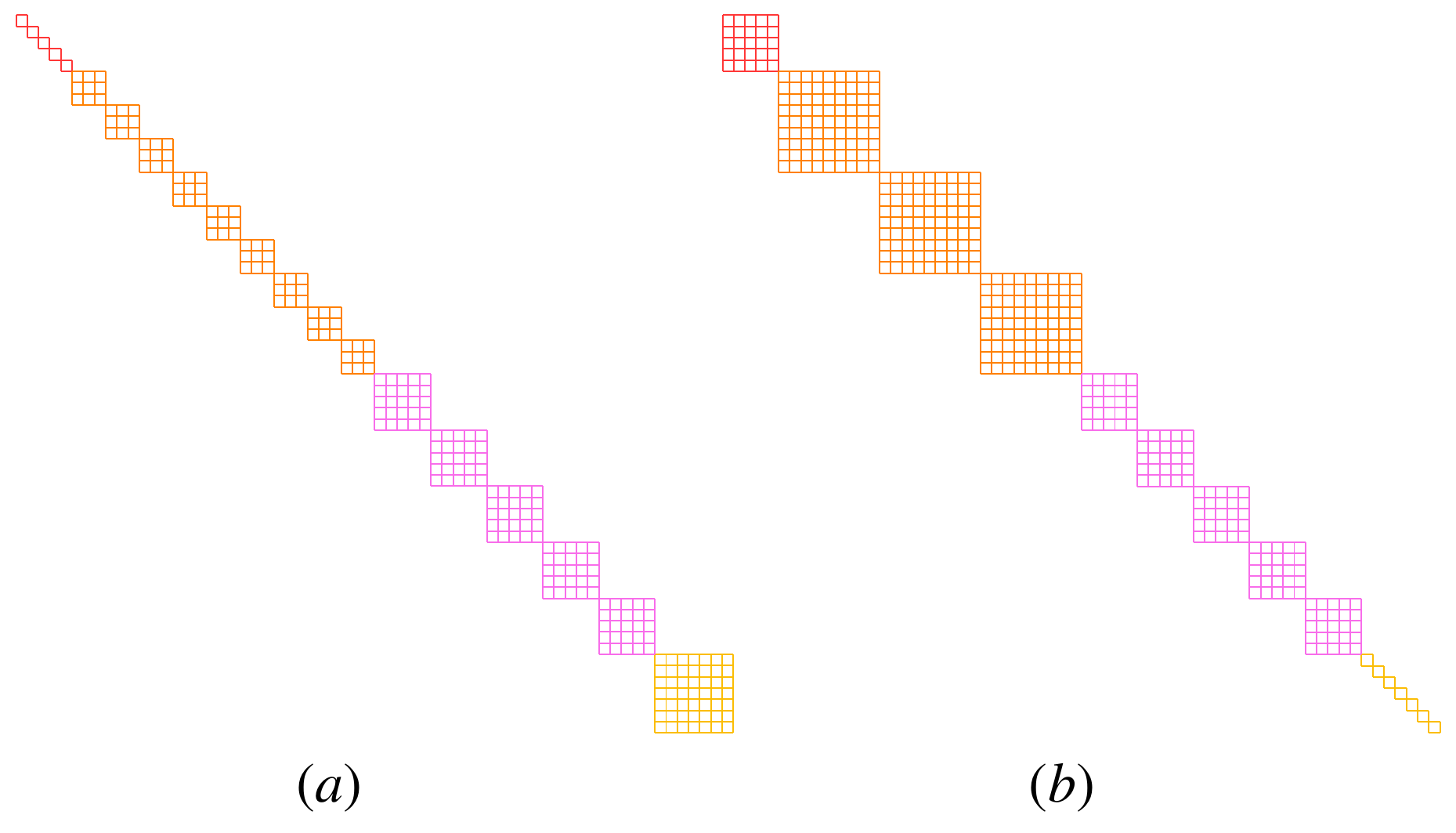}
    \caption{\((a)\) is the decomposition of \((\mathbb{C}^2)^6\) with respect to \(\operatorname{SU}(2)\), while \((b)\) is for \(S_6\) by Schur--Weyl duality.}
    \label{SchurWeyl}
\end{figure}

By definition, one can check that \( \pi_{S_n} \) commutes with \( \pi_{\SU(d)} \). The so-called double commutant theorem (see \cite{Sagan01,Goodman2009,Tolli2009} and the SM) strengthens the statement that any matrix that commutes with \( \pi_{\SU(d)} \) is a representation of \( \mathbb{C}[S_n] \), and vice versa. This fact has been mentioned many times before, since \( \SU(d) \)-symmetric Hamiltonians are completely captured by \( \mathbb{C}[S_n] \). Moreover, let \( f_1, f_2 \in \mathbb{C}[S_n] \). Their convolution is defined as
\begin{align}
    (f_1 \ast f_2)(\sigma) \vcentcolon = \sum_{\tau \in S_n} f_1(\tau) f_2(\tau^{-1} \sigma),  
\end{align}
which is again an element of \( \mathbb{C}[S_n] \). An \( S_n \)-equivariant convolution is established in the following identity by the left action \( L_\eta \) on \( \mathbb{C}[S_n] \): 
\begin{align}\label{EquivConv}
    \big( (L_\eta f_1) \ast f_2 \big)(\sigma) = \big(L_\eta (f_1 \ast f_2) \big)(\sigma).
\end{align}
The representation of the left action \( L_\eta \) is just \( \pi(\eta) \), and Eq.~\eqref{EquivConv} reflects the associativity of multiplication or, equivalently, the commutativity of left and right translation actions of operators on the representation space or Fourier space, in the sense that \( \wt{\pi}(f) \) encodes the Fourier coefficients of \( f \) \cite{KondorTrivedi2018,kondor2}. Fourier-space CNNs have been applied to spherical image recognition \cite{Cohen2017spherical}, chemistry \cite{TensorFieldNetworks2018arxiv,CormorantNeurIPS}, and particle physics \cite{Kondor_2020} (e.g., Lorentz group equivariance). 

It should be noted that the \( \SU(d) \) tensor product representation also commutes with \( S_n \) representations. To avoid ambiguity, it is referred to as the \( \SU(d) \) symmetry, while the \( S_n \)-equivariance mentioned here is defined in the sense of Eq.~\eqref{EquivConv} using convolutions.


\section{PQP+ and Super-polynomial Speedup}

A unitary transformation that decomposes \( V^{\otimes n} \) as above is called the \emph{quantum Schur transform} (QST), denoted \( U_{\operatorname{Sch}} \). Based on the context, the Schur transform can be theoretically conducted in two ways. The first uses sequential Clebsch–Gordan decomposition \cite{harrow2005PhD, pqc, sergii1}, which transforms standard matrix representations of \( \SU(d) \) into irrep matrix blocks, as shown in Fig.~\ref{SchurWeyl}~(a). When \( d = 2 \), this method is common in quantum physics, where total angular momentum or spin is used to label \( \SU(2) \) irreps. The second method relies more on \( S_n \) representation theory \cite{Tolli2009,krovi,Zheng2021SpeedingUL} and appears less intuitive at first glance. It decomposes permutation matrices into \( S_n \) irrep matrix blocks, as shown in Fig.~\ref{SchurWeyl}~(b). The relationship between these two methods is discussed in \cite{Zheng2021SpeedingUL}. As the main focus of this work is on quantum algorithms with \( \SU(d) \) symmetry, we will work with \( S_n \) irreps in the following. We call the basis that reveals \( S_n \) irreps the \emph{Young–Yamanouchi basis}, or simply the Young basis. For more details on Schur–Weyl duality and \( S_n \) representation theory, we refer readers to Refs.~\cite{Sagan01,Goodman2009,Tolli2009}.

In the original setting of PQC, we first prepare a Young basis vector \( \ket{v} = U_{\operatorname{Sch}} \ket{x} \) for some \( x \in \{0, 1\}^n \). The QST \( U_{\operatorname{Sch}} \) can be efficiently implemented for qudits with gate complexity \( \operatorname{Poly}(n, \log d, \log(1/\epsilon)) \) \cite{harrow2005PhD,Harrow06,krovi}. We then measure the matrix entries \( \bra{u} \pi(\sigma) \ket{v} \). Any permutation can be synthesized by \( \mathcal{O}(n^2) \) SWAP gates, and these measurements define the permutational quantum polynomial (PQP) class \cite{pqc}. The ability to compute the matrix entries of the \( S_n \) irreps in polynomial time was conjectured to reflect non-classical aspects of quantum computing. One motivation for this conjecture is the scaling behavior of the dimensions of irreps arising from \( V^{\otimes n} \) via Schur–Weyl duality. Given a Young diagram \( \lambda \) with at most \( d \) rows,
\begin{align}
	& \dim W_\lambda = \prod_{1 \leq i < j \leq d} \frac{\lambda_i - \lambda_j - i - j}{i - j}, \\
	& \dim S^\lambda = \frac{n!}{\prod_{i,j} h_{i,j}},
\end{align}
where the latter is the \emph{hook-length formula}, and \( h_{i,j} \) denotes the number of boxes including the \( (i,j) \) box, plus those to the right and below it \cite{Sagan01,Tolli2009}. It is easy to check that \( \dim W^\lambda \) is bounded polynomially in \( n \), while \( \dim S^\lambda \) may scale exponentially depending on \( \lambda \). For example, let \( n = 2m \) (even) and \( d = 2 \), and take \( \lambda = (m, m) \). Then
\begin{align}
	\dim S^\lambda = \frac{(2m)!}{(m+1)! \, m!} = \frac{2^m}{m+1} \prod_{k=1}^m \frac{2k-1}{k} > \frac{2^m}{m+1}.
\end{align}
Evaluating matrix entries of this size was thought to be classically hard, but a classical sampling method was later found in \cite{sergii1,sergii2} for qubits.

In this work, we extend the PQP class to a larger class, PQP\(^+\), to better capture non-classical features. As mentioned earlier, the motivation comes from interpreting quantum circuits as Fourier space representations of groups. Let $f = \sum_{\sigma_i \in S_n} c_i \sigma_i \in \mathbb{C}[S_n]$. That is, $f(\sigma_i) = c_i$. We further assume that its matrix representation $\wt{\pi}(f) = \sum_{\sigma_i \in S_n} c_i \, \pi(\sigma_i)$ is Hermitian. We define PQP\(^+\) as the set of polynomially-bounded quantum algorithms capable of approximating matrix entries of the following Hamiltonian simulation on qudits:
\begin{align} \label{autocorr}
	\exp(-i \wt{\pi}(f)) & = \sum_{n = 0}^\infty \frac{(-i)^n}{n!} (\wt{\pi}(f))^n \\
	& = \sum_{n = 0}^\infty \frac{(-i)^n}{n!} \wt{\pi}(\underbrace{f \ast \cdots \ast f}_{n \text{ times}}). \notag
\end{align}
By the Wedderburn theorem, Eq.~\eqref{autocorr} encompasses all \emph{unitary} \( S_n \)-Fourier coefficients. If \( f \) is supported on \( k \) qubits, it can be efficiently simulated using standard Hamiltonian simulation methods. We show that the implementation complexity, which depends linearly on time and polynomially on the number of qudits, can be achieved using the Linear Combination of Unitaries (LCU) method \cite{Childs2012,Berry_2015,Childs2018}, as formalized in the following theorem. Combined with efficient QST for preparing Young basis vectors, this generalization of PQC to \( \mathbb{C}[S_n] \) shows the potential for \emph{super-polynomial} quantum speedup in computing matrix entries of \( \mathbb{C}[S_n] \) Fourier coefficients, for which no known classical polynomial-time algorithm exists.

\begin{theorem} \label{thmpqcd}
	Let $H = \wt{\pi}(f) = \sum_i c_i \sigma_i$ be a $k$-local $\mathbb{C}[S_n]$ Hamiltonian with Young basis elements $\ket{u}, \ket{v}$ prepared using the efficient QST. The quantum circuit on qudits is able to simulate the matrix entry $\bra{u}\exp(-itH) \ket{v}$ by 
	\begin{align}
		\mathcal{O}\left( tCk^3 n^k \frac{\log (tC kn^k/\epsilon)}{\log \log (tC kn^k/\epsilon)}\right)
	\end{align}
	SWAP gates with the constant $C = \max_i \vert c_i \vert$. 
\end{theorem}
\begin{proof}
	We first note a standard result for the symmetric group $S_n$: the number of permutations supported exactly on $l$ elements, or equivalently with $n-l$ elements being fixed is given by the \textit{derangement}:
	\begin{align}
		D_l = \frac{n!}{(n-l)!} \Big( \frac{1}{2!} - \frac{1}{3!} + \cdots + (-1)^l \frac{1}{l!} \Big).
	\end{align}
	If $H$ is $k$-local, then it contains at most $D_2 + \cdots + D_k$ different permutations in $S_n$. Since $n!/ (12 (n-l)!)  \leq D_l \leq n!/ (4 (n-l)!)$, the number is of order $\mathcal{O}(k n^k)$. Let $C = \max_i \vert c_i \vert$, then $\Vert c \Vert_1 \vcentcolon = \sum_i \vert c_i \vert \leq \mathcal{O}( C k n^k)$ where $\Vert c \Vert_1$ is defined in the setting of LCU.
	
	Let us divide the Hamiltonian evolution $\exp(-itH)$ into $M$ steps $\exp(-i \Delta t H)$. We set $M = t C k n^k$ so that $\Delta t \Vert c \Vert_1 = t \Vert c \Vert_1 /M = \mathcal{O}(1)$. This is a crucial step to validate the so-call \emph{oblivious amplitude amplification} in LCU (see \cite{Childs2012,Berry_2015,Childs2018} and SM for more details). Then we truncate the Taylor series of each product term to order $K$:
	\begin{align}
		\hspace{-3mm} \Vert \exp(-i \Delta t H) - \hspace{-1mm} \sum_{m = 0}^{K-1} \frac{(-i\Delta t H)^m}{m!} \Vert = \frac{(\Delta t \Vert H \Vert )^K}{K!} \leq \tilde{\epsilon}. \hspace{-1mm}
	\end{align}
	In order to bound the total error by $\epsilon$, each product term should be simulated within error $\tilde{\epsilon} = \epsilon/M$. Using $(m/e)^m \leq m!$, the above inequality holds when $K \log (K/\Delta t \Vert H \Vert e) \geq \log (1/\tilde{\epsilon})$. Note that $\Vert H \Vert$ means operator norm here and should not be confused with $\Vert c \Vert_1$. However, by definition,
	\begin{align}
		\Vert H \Vert = \max_{\Vert \ket{\psi} \Vert = 1} \Vert H \ket{\psi} \Vert \leq \sum_i \vert c_i \vert \max_{\Vert \ket{\psi} \Vert = 1} \Vert \sigma_i \ket{\psi} \Vert. 
	\end{align}
	As $\sigma_i$ are unitary, $\Vert H \Vert$ is upper bounded by $\Vert c \Vert_1$. Thus $\Delta t \Vert H \Vert = \mathcal{O}(1)$, we can further relax the condition to require $K \log K \geq \log (1/\tilde{\epsilon})$. Let $K'  \geq \frac{\log 1/\tilde{\epsilon}}{\log \log 1/\tilde{\epsilon}}$. Then
	\begin{align}
		K' \log K' \geq \log \frac{1}{\tilde{\epsilon}} - \log \frac{1}{\tilde{\epsilon}} \frac{\log \log 1/\tilde{\epsilon}}{\log \log \log 1/\tilde{\epsilon}} \approx  \log \frac{1}{\tilde{\epsilon}}  
	\end{align}
	for small $\tilde{\epsilon}$. Setting $K = K'$ this bound also recovers the result of \cite{Berry_2015}. Since the time evolution also contains $M$ steps and since any $k$-local permutation $\sigma_i$ can be written in $\mathcal{O}(k^2)$ geometrically local SWAPs, the total circuit complexity is: 
	\begin{align}
		\mathcal{O}\left( k^2 MK \right) = \mathcal{O}\left(tCk^3 n^k \frac{\log (tC kn^k/\epsilon)}{\log \log (tC kn^k/\epsilon)}\right),
	\end{align}
	which finishes the proof.
\end{proof}

Note that SWAP gates on qudits are assumed to be easy to implement with constant overhead. We can thus disregard the polynomially-scaling constant on $k$ in the approximation and write the total circuit complexity $\mathcal{O}( t n^k \log (tn^k/\epsilon)/{\log \log (tn^k/\epsilon)}$. Theorem \ref{thmpqcd} establishes a generic possible quantum {super-polynomial} speedup. 

In qubits, the permutation operators are related by the exchange interactions. 
\begin{align} \label{trans}
	\pi((i \: j)) = 2\hat{\boldsymbol{S}}_{i} \cdot \hat{\boldsymbol{S}}_{j} + \frac{1}{2}I,
\end{align}
where $\hat{\boldsymbol{S}}_{i} = \frac{1}{2}(X,Y,Z)$ is defined by the Pauli operators and the identity is simply obtained by the expansion of SWAPs under Pauli basis \cite{Heisenberg1928, Klein_1992,RobertsChaos2017}. For example, consider $f = (1 \: 2) + (2 \: 3) + (3 \: 4) + (4 \: 1)$. Therefore under the representation $\wt{\pi}$,
\begin{align}
\begin{aligned}
	& H_P \equiv \wt{\pi}(f) \\
	= & 2(\hat{\boldsymbol{S}}_{1} \cdot \hat{\boldsymbol{S}}_{2} + \hat{\boldsymbol{S}}_{2} \cdot \hat{\boldsymbol{S}}_{3} + \hat{\boldsymbol{S}}_{3} \cdot \hat{\boldsymbol{S}}_{4} + \hat{\boldsymbol{S}}_{4} \cdot \hat{\boldsymbol{S}}_{1} + I), 
\end{aligned}
\end{align}
is simply the 1D Heisenberg chain with a periodic boundary condition of 4 spins. By definition, the standard Heisenberg model admits the $\SU(2)$ symmetry and has received extensive studies in quantum many body theory \cite{Marshall1955,Lieb62,Tasaki2020}. Analyzing the ground state property of Heisenberg Hamiltonian using symmetry adapted variational quantum algorithm is also developed in, e.g., Refs.~\cite{VieijraRBM,Seki_2020,Zheng2021SpeedingUL}.

\begin{corollary}\label{thmpqc}
	With the above assumption, but exclusively on a qubits system, the matrix entry $\bra{u}\exp(-itH) \ket{v}$ can be simulated by a quantum circuit with
	\begin{align}
		\mathcal{O}\left( t L(k)n^{k} \frac{\log(tL(k)n^k/\epsilon)}{\log\log(tL(k)n^k/\epsilon)}\right)
	\end{align}
	$k$-local Pauli operators with the constant $L(k) = 2^{k-1}C$.
\end{corollary}
\begin{proof}
	We first note a simple fact that any $\sigma \in S_n$ can be expanded as a product of cycles, e.g., $(123) = (12)(23) \in S_3$. By assumption, $\sigma_i$ is $k$-local and hence can be expanded by at most $k-1$ transpositions $\tau_{i_j} = (i_j, i_{j+1})$. These transpositions may not be geometrically local, but by Eq.\eqref{trans} one can compile them as a product of Pauli gates. Then $\sigma_i$ equals
	\begin{align} \label{ham_Pauli}
		& \prod_{j = 1}^{k-1} \tau_{i_j} = \prod_{j = 0}^{k-1} (2\hat{\boldsymbol{S}}_{i_j} \cdot \hat{\boldsymbol{S}}_{i_{j+1}} + \frac{1}{2} I) \\
		= & \sum_{j = 0}^{k-1} \binom{k-1}{j} 2^{k-1-2j} \underbrace{( \hat{\boldsymbol{S}}_{i_1} \cdot \hat{\boldsymbol{S}}_{i_2}) \circ \cdots \circ ( \hat{\boldsymbol{S}}_{i_{k-1}} \cdot \hat{\boldsymbol{S}}_{i_k}  ) }_{j \operatorname{ \text{many terms replaced by}} I}. \notag
	\end{align}
	The operator norm of $\hat{\boldsymbol{S}}_{i_j} \cdot \hat{\boldsymbol{S}}_{i_{j+1}}$ can be seen directly from Eq. \eqref{trans} as $\frac{3}{4}$. Because we are going to expand $H$ by Pauli operators and then apply LCU, we should recompute $\Vert c \Vert_1$ used in Theorem \ref{thmpqcd}. As a first step, the sum of absolute values of coefficients of each $\sigma_i$ is bounded by
	\begin{align}
		\sum_{j = 0}^{k-1} 2^{k-1-2j} (\frac{3}{4})^{k-1-j} \binom{k-1}{j} = 2^{k - 1},
	\end{align}
	where the last step uses binomial theorem. Then $\Vert c \Vert_1$ is bounded by $n^k L(k)$ as $H$ contains at most $\mathcal{O}(n^k)$ $k$-local permutations. Substituting this into Theorem \ref{thmpqcd}, we complete the proof.
\end{proof}

A more involved derivation using Pauli gates in Theorem \ref{thmpqc} in the form of unitary coupled clusters suggests a potentially hardware-friendly implementation on quantum computers \cite{Anand_2021}.  

Our results suggest that the PQP+ class is able to capture some non-classical aspects of quantum computation in the qudit case. The above form of $\mathbb{C}[S_n]$ Hamiltonian simulation by Eq.~\eqref{trans} encompasses a large class of physically relevant problems, such as the Heisenberg spin model. Naturally, multiplication of $\mathbb{C}[S_n]$ Hamiltonian operators correspond to $\mathbb{C}[S_n]$ convolution, which leads to further useful applications of machine learning and optimization tasks. 
In the related work \cite{Zheng2021SpeedingUL}, we focus on utilizing the quantum speedup to address quantum machine learning tasks by designing convolutional quantum ans{\"a}tze. From this point of view, PQP+ may be interpreted as the natural complexity class for \emph{quantum $S_n$-Fourier space activation} for quantum machine learning and optimization tasks. 


\section{Discussion}

The significance of studying PQC originally stems from imitating the behavior of anyons by a quantum computer. The so-called braided group $B_n$ which determines the symmetry of these systems are relegated to its elementary counterpart: the symmetric group $S_n$ \cite{pqc,Kitaev2006}. Nevertheless, studying matrix entry from $\mathbb{C}[S_n]$ is still a highly nontrivial task for testing quantum advantages. 

Theoretically, Schur--Weyl duality indicates those operators who commute with $\SU(d)$ tensor product representation, i.e., with $\SU(d)$ symmetry, must be elements from $\mathbb{C}[S_n]$ under the permutation representation and vice versa \cite{Goodman2009,Tolli2009}. This enriches the context of PQP+ with at least two more research directions: $(a)$ one could study $\SU(d)$ symmetric Hamiltonian. For instance, the spin-$\frac{1}{2}$ Heisenberg Hamiltonian $H$ with or without frustration admits the $\operatorname{SU}(2)$ invariance. Schur--Weyl duality indicates that matrix representation of $H$ under Young basis looks like what we draw in Fig.\ref{SchurWeyl} $(b)$. Even the dimension of $S_n$-irreps would still scale exponentially, it is verified numerically in Ref.~\cite{Zheng2021SpeedingUL} that the ans{\"a}tze preserving $\operatorname{SU}(2)$ symmetry, or what we call $S_n$-convolutional quantum ans{\"a}tze ($S_n$-CQA), exhibit a high performance in searching ground states. $(b)$ one could study quantum circuits with $\SU(d)$ symmetry as imposing symmetry would make some well-known results for common quantum circuits false and provide more deep insights into quantum computation theoretically. For instance, any quantum circuit can be generated by 2-local unitaries for qubits as well as for qudits \cite{Vlasov2001,Brylinski2001,Sawicki2016}. However, this statement is proved to be false in Refs.~\cite{MarvianNature,MarvianSU2,Zheng2021SpeedingUL} for $\U(1)$ and $\SU(d)$ symmetric quantum circuits. Symmetric unitaries with trivial relative phases can be generated locally, while replenishing a complete basis for relative phases necessitate the non-locality. The symmetric universality leads to the construction of symmetric random circuits, which is desirable for a physically concrete model of near optimal covariant quantum error-correcting code \cite{Zhou2021newperspectives,kong2022near}. In addition, it also relates to recent works on the linear growth of circuit complexity \cite{brandao2021models,brian2022linear, haferkamp2022linear, li2022short}. Furthermore, by replacing the scrambling of $\U(1)$-charged information \cite{yoshida:soft,junyu2020chargescrambler} with generalized SU($d$)-charged information, we may also consider the $\SU(d)$-symmetric Hayden-Preskill protocol. 

On practical sides, Permutation-equivariant layers from $\mathbb{C}[S_n]$ admit a hardware-friendly compilation based on exchange-type interactions. Exponential‑SWAP (\(e^{-i\theta\,\mathrm{SWAP}}\)) and related \(\sqrt{\mathrm{SWAP}}\)/$i$SWAP gates, together with single‑qubit rotations, generate the required symmetric operations and are universal on qubits \cite{divincenzo2000universal}, avoiding an explicit Schur transform in the critical path. As such, symmetry-preserving variational circuits are compatible with several platforms. Trapped ions and neutral‑atom arrays offer global rotations and collective spin–spin couplings for preparing symmetric (Dicke) superpositions; superconducting and bosonic (cavity/circuit QED) devices provide native iSWAP/XY/parametric couplers and eSWAP-style entanglers \cite{gao2018entangling}. Moreover, symmetry can be directly leveraged for error mitigation in near-term experiments. In particular, \emph{symmetry verification} projects noisy outputs back to the target symmetry sector and has been validated experimentally in variational settings~\cite{BonetMonroig2018SymmetryVerification,Sagastizabal2019SVExperiment}. Beyond the mitigation, families of permutation-invariant codes \cite{Ouyang2014PIcodes, Ouyang2021Deletion} and $\U(1)$- ($\SU(d)$-)covariant codes \cite{Zhou2021newperspectives,kong2022near} provide natural candidates of quantum error correction protocols.

\vspace{3 mm}

\section*{Acknowledgements}
HZ and ZL contribute equally in this work. We thank Liang Jiang, Antonio Mezzacapo, Kristan Temme,  Hy Truong Son, Erik H. Thiede, Pei Zeng, and Changchun Zhong for valuable discussions.  JL is supported in part by International Business Machines (IBM) Quantum through the Chicago Quantum Exchange, and the Pritzker School of Molecular Engineering at the University of Chicago through AFOSR MURI (FA9550-21-1-0209). JL is also supported in part by the University of Pittsburgh, School of Computing and Information, Department of Computer Science, Pitt Cyber, PQI Community Collaboration Awards, John C. Mascaro Faculty Scholar in Sustainability, NASA under award number 80NSSC25M7057, and Fluor Marine Propulsion LLC (U.S. Naval Nuclear Laboratory) under award number 140449-R08. SS acknowledges support from the Royal Society University Research Fellowship. This research used resources of the Oak Ridge Leadership Computing Facility, which is a DOE Office of Science User Facility supported under Contract DE-AC05-00OR22725.

\emph{Note added:} Our final revisions coincided with the emergence of numerous interesting papers on quantum circuits with symmetries. This new literature covers a broad range of topics, including geometric quantum machine learning \cite{ragone2022representation,sauvage2022building}, thermalization \cite{majidy2023critical,Majidy2023review,Liu2024Mpemba,chang2024deep,SUd-k-Design2023Application}, symmetric random circuits \cite{SUd-k-Design2023,U(1)Design2023,Marvian3local,mitsuhashi2024Designs2,li2024efficient}, complexity \cite{castrosilva2025symmetricquantumcomputation}, and approximate error correction \cite{PhysRevLett.133.240603,lin2024covariantquantumerrorcorrectingcodes}. We have addressed most of these topics in our Discussion, and their promising progress highlights the growing importance of symmetric quantum systems at the intersection of physics, mathematics, and quantum information.

\bibliography{references.bib}

\clearpage
\widetext
\appendix
\begingroup
\section*{Appendix}


\section{More facts about Schur--Weyl duality and $S_n$ representation}

\begin{lemma}\label{HeisenbergSWAP}
	Any local part $\hat{\boldsymbol{S}}_{i} \cdot \hat{\boldsymbol{S}}_{j}$ constituting a Heisenberg Hamiltonian $H = \sum_{ij} J_{ij} \hat{\boldsymbol{S}}_{i} \cdot \hat{\boldsymbol{S}}_{j}$ can be written as
	\begin{align}
		\hat{\boldsymbol{S}}_{i} \cdot \hat{\boldsymbol{S}}_{j} = \frac{1}{2} (ij) - \frac{1}{4} I.
	\end{align}
\end{lemma}
\begin{proof}
	Let us consider $\hat{\boldsymbol{S}}_1 \cdot \hat{\boldsymbol{S}}_2$. Expanded by definition,
	\begin{align}
		\hat{\boldsymbol{S}}_1 \cdot \hat{\boldsymbol{S}}_2 = \frac{1}{2} (J_{12}^2 - J_1^2 - J_2^2),
	\end{align}
	where for now the subscripts on $J^2$ denotes all sites $J^2$ acting on. Under total spin basis, it is easy to see that
	\begin{align}
		J_{12}^2 - J_1^2 - J_2^2 = \begin{pmatrix} 2 & 0 & 0 & 0 \\ 0 & 2 & 0 & 0 \\ 0 & 0 & 2 & 0 \\ 0 & 0 & 0 & 0  \end{pmatrix} - \frac{3}{4} I - \frac{3}{4} I = \begin{pmatrix} \frac{1}{2} & 0 & 0 & 0 \\ 0 & \frac{1}{2} & 0 & 0 \\ 0 & 0 & \frac{1}{2} & 0 \\ 0 & 0 & 0 & -\frac{3}{2} \end{pmatrix}.
	\end{align}
	While
	\begin{align}
		(12) = \begin{pmatrix} 1 & 0 & 0 & 0 \\ 0 & 1 & 0 & 0 \\ 0 & 0 & 1 & 0 \\ 0 & 0 & 0 & -1 \end{pmatrix}. 
	\end{align}
	Therefore,
	\begin{align}
		\hat{\boldsymbol{S}}_1 \cdot \hat{\boldsymbol{S}}_2 = \frac{1}{2} (J_{12}^2 - J_1^2 - J_2^2) = \frac{1}{2} ((12) - \frac{1}{2}I).
	\end{align} 
	This argument holds for any $i,j$, hence the proof follows. 
\end{proof}

\begin{theorem}\label{Commutant}
	Let $\langle \pi_{\operatorname{SU}(d)} \rangle$ be the collection of all matrices on $V^{\otimes n}$ generated by $\pi_{\operatorname{SU}(d)}$ (i.e., taking both linear spans and matrix products) and let $\langle \pi_{S_n} \rangle$ be defined similarly. Note that $\langle \pi_{S_n} \rangle$ is just the representation of $\mathbb{C}[S_n]$. Let $M$ be an arbitrary matrix of $V^{\otimes n}$. If it commutes with $\langle \pi_{\operatorname{SU}(d)} \rangle$, then $M \in \langle \pi_{S_n} \rangle$. If it commutes with $\langle \pi_{S_n} \rangle$, then $M \in \langle \pi_{\operatorname{SU}(d)} \rangle$.
\end{theorem}

This theorem is proved by the so-called \textit{double commutant theorem}, details can be found in \cite{Goodman2009}. 

\begin{theorem}[Wedderburn Theorem]\label{Wedderburn}
	Given any $S_n$-irrep $S^\lambda$, let $\operatorname{End}(S^\lambda)$ denote the collection of linear transformation of $S^\lambda$. With respect to any basis, e.g., the Young basis, $\operatorname{End}(S^\lambda)$ is simply the collection of all $\dim S^\lambda \times \dim S^\lambda$ matrices. As a vector space, the group algebra $\mathbb{C}[S_n]$ is isomorphic with the direct sum of $\operatorname{End}(S^\lambda)$:
	\begin{align}
		\mathbb{C}[S_n] \cong \bigoplus_{\lambda \vdash n} \operatorname{End}(S^\lambda).
	\end{align}
\end{theorem}

We emphasize that being different from decomposing $V^{\otimes n}$ by $\operatorname{SU}(d)-S_n$ duality, all kinds of Young diagrams $\lambda$, standing for inequivalent $S_n$-irreps, appears once and only once in the above direct sum decomposing $\mathbb{C}[S_n]$. In any case, restricting to each $\lambda$, $\wt{\pi}_{\mathbb{C}[S_n]}$ produces all $\dim S^\lambda \times \dim S^\lambda$ matrices and hence any matrix commuting with $\wt{\pi}_{\mathbb{C}[S_n]}$ should be a scalar.


\section{Linear Combination of Unitaries}

Linear Combination of Unitaries (LCU) was proposed in \cite{Childs2012,Berry_2015,Childs2018} and we adapt this as a possible methods for simulating Hamiltonians constructed by elements from $\mathbb{C}[S_n]$. With the awareness of Schur--Weyl duality, these Hamiltonians have direct physics implication: they preserve $\operatorname{SU}(d)$ symmetry. A well-known example is the Heisenberg Hamiltonian of magnetism \cite{Heisenberg1928,Marshall1955,Lieb62,Tasaki2020}. We provide in the following a brief description of LCU. Practical implementation details with complexity analysis can be found in \cite{Childs2012,Berry_2015,Childs2018}. 

Suppose we are going to simulate $e^{-itH} \ket{\psi}$ for some state $\ket{\psi}$ with $H = \wt{\pi}(f)$ for some $f \in \mathbb{C}[S_n]$. As we written in the main text, $e^{-itH}$ would always be expanded and truncated to a certain order $K$. Let us denote the approximate matrix by $M$ and suppose $M = \sum_{j = 1}^m \alpha_j V_j$ with $V_j$ being unitaries. Note that in our case, $V_j$ are products of permutations $\sigma_{i_1},...,\sigma_{i_K} \in S_n$. The coefficients $\alpha_j$ can be set as positive reals, as their complex phase can be absorbed into $V_j$. We denote by $\Vert \alpha \Vert_1 = \sum_j \vert \alpha_j \vert$ for later use.

Then we need to prepare $\lceil \log_d m \rceil$ ancilla qudits with respect to which the computational basis is registered as $\ket{j}$ for $j = 0,...,m$. Let $W$ be a unitary acting on these ancillas such that 
\begin{align}
	W \ket{0} = \frac{1}{\sqrt{\Vert \alpha \Vert_1}} \sum_j \sqrt{\alpha_j} \ket{j}.
\end{align}
Assume we can also construct control operators: $\ket{j} \bra{j} \otimes V_j$ with $V = \sum_{j=1}^m \ket{j} \bra{j} \otimes V_j$. Then
\begin{align}
	(\bra{0} \otimes I) W^{-1} V W \ket{0} \ket{\psi} =  \frac{1}{\Vert \alpha \Vert_1} \Big( \sum_j \sqrt{\alpha_j} \bra{j} \otimes I \Big)  \Big( \sum_j \sqrt{\alpha_j} \ket{j} V_j \ket{\psi} \Big) = \frac{1}{\Vert \alpha \Vert_1} M \ket{\psi},
\end{align}
which gives the state we want to approximate when we measure ancillas with the outcome being $0$. However, the probability to retrieve this outcome is $\Vert M\ket{\psi} \Vert^2 / \Vert \alpha \Vert_1^2$ which may not be large enough, Then we apply \emph{oblivious amplitude amplification} to amplify it. Roughly speaking, by applying $-U R U^{-1} R$ with $R \vcentcolon = (I - 2\ket{0} \bra{0}) \otimes I$ and $U \vcentcolon = W^{-1} V W$ to $\ket{0} \ket{\psi}$, the amplitude of the target state will be enlarged. Note that as $M \approx e^{-itH}$ is nearly unitary, $\Vert M\ket{\psi} \Vert^2 / \Vert \alpha \Vert_1^2 \approx 1// \Vert \alpha \Vert_1^2$ and $\mathcal{O}(\frac{1}{\Vert \alpha \Vert_1^2})$ times amplification will be enough \cite{Childs2012,Berry_2015,Childs2018}.

In our case, let $H = \wt{\pi}(f)$ with $f = \sum_i c_i \sigma_i$ and $\sigma_i$ being $k$-local. Then
\begin{align}
	\exp(-i \Delta t H) = \sum_{m=0}^K \sum^N_{i_1,..., i_m} \frac{(-i\Delta t)^m}{m!}c_{i_1}...c_{i_m} \sigma_{i_1} \cdots \sigma_{i_m} = \sum^K_{m=0} \sum^N_{i_1,..., i_m} \wt{c}_{{i_1}...{i_m}} \wt{\sigma}_{i_1,... i_m} + \mathcal{O}(\tilde{\epsilon}) = M + \mathcal{O}(\tilde{\epsilon}),
\end{align}
where $N = O(n^k)$ is the number of all $k$-local permutations. To apply LCU, we need $O(K \log_d N)$ qudits. We divide the evolution into small steps with time $\Delta t$ in each step such that
\begin{align}
	\Vert \tilde{c} \Vert_1 = \sum_{k = 0}^K \sum_{i_1,...,i_m}^N c_{i_1...i_m} \leq \sum_{k = 0}^\infty \frac{(-i \Delta t)^m}{m!} \sum_{i_1,...,i_m} c_{i_1} \cdots c_{i_m} = e^{\Delta t \Vert c \Vert_1} = O(1).
\end{align}
Thus applying the oblivious amplitude amplification by constant times is sufficient to complete the task.


\end{document}